\documentclass[submission,copyright,creativecommons]{eptcs}
\providecommand{\event}{Non-Classical Logics 2024 (NCL'24)}

\usepackage{iftex}

\usepackage{graphicx}
\usepackage{llncsdoc}
\usepackage{setspace}
\usepackage{amsmath,amssymb}
\usepackage{amsthm}
\usepackage{proof, amssymb, my-macro2}

\theoremstyle{definition}
\newtheorem{dfn}{Definition}[section]
\newtheorem{lemma}[dfn]{Lemma}
\newtheorem{theorem}[dfn]{Theorem}

\ifpdf
  \usepackage{underscore}         
  \usepackage[T1]{fontenc}        
\else
  \usepackage{breakurl}           
\fi

\title{Nested-sequent Calculus for Modal Logic MB}
\author{Tomoaki Kawano
\institute{Kanagawa University}
\email{kawano.t.af@m.titech.ac.jp}
}
\def\titlerunning{Nested-sequent Calculus for Modal Logic MB}
\def\authorrunning{Tomoaki Kawano}
\begin{document}
\maketitle

\begin{abstract}
Quantum logic ({\bf QL}) is a non-classical logic for analyzing the propositions of quantum physics.
Modal logic {\bf MB}, which is a logic that handles the value of the inner product that appears in quantum mechanics, was constructed with the development of {\bf QL}.
Although the basic properties of this logic have already been analyzed in a previous study, some essential parts still need to be completed. They are concerned with the completeness theorem and the decidability of the validity problem of this logic.
This study solves those problems by constructing a nested-sequent calculus for {\bf MB}. In addition, new logic {\bf MB+} with the addition of new modal symbols is discussed.

\end{abstract}

\section{Introduction}
Quantum logic ({\bf QL}) has developed from both quantum physics and mathematical logic aspects since \cite{birk}.
{\it Modular lattices} and {\it orthomodular lattices} have been analyzed as algebraic semantics of {\bf QL}. These lattices are based on a {\it Hilbert space}, which is the state space of a particle.
In quantum mechanics, the value of a physical quantity can only be predicted probabilistically.
The absolute value of the {\it inner product} of two states (two unit vectors in a Hilbert space) is intrinsically related to the probability distribution of the physical quantity.

As counterparts of orthomodular lattice, some {\it Kripke frames} (binary relation frames) have also been analyzed.
In the simplest Kripke frame of {\bf QL},
possible worlds represent states,
and the binary relation abstractly 
represents the {\it orthogonal relation} between states.
Intuitively, on this frame, we can only deal with the binary concept of whether a proposition is 100 \% true or not because the orthogonal relation expresses that the inner product between states is zero.
Although such logic has developed as an essential foundation for {\bf QL}, developing logic that can handle detailed probability values is also desirable. 
Because the absolute value of the inner product is independent of the order of the elements, the binary relation is constructed to satisfy {\it symmetry} in these frames.

 {\it Extended quantum logic} ({\bf EQL}) \cite{extendQL} has been developed to handle some properties of the absolute value of the inner product.
In \cite{extendQL}, two logics, {\bf EQL} and {\bf MB}, are constructed. 
The truth values of the formulas of {\bf EQL}
range over the unit interval $I=[0,1]$, which is related to the absolute value of the inner product. 
{\bf MB} (multi-modal extension of {\bf B}) is the modal logic counterpart of {\bf EQL}.
This relation could be regarded as the well-known McKinsey–Tarski translation.
In {\bf MB}, the truth value is binary, but the concept of the inner product can be expressed using a modal symbol containing numerical values.
This study focuses on {\bf MB}.

Technically, as a relation between states,  we can also consider frames that introduce not the absolute value of the inner product but the inner product itself.
However, when analyzing the critical factor of probability, a frame that introduces the inner product itself becomes somewhat unnecessarily complex.
Therefore, the study of {\bf MB} deals with frames that introduce only absolute values \cite{extendQL}.
Other studies have introduced the {\it transitions} between two states in Hilbert space as a binary relation of the frame.
For example, the frame of {\it dynamic quantum logic} introduces the concepts of {\it unitary transformations} and {\it projections} \cite{bat1}. 
Each of these has its logical characteristics and has been studied separately. 


Although the basic concept of {\bf MB} has already been analyzed in \cite{extendQL},
there is room for analysis of the following concepts:

\begin{enumerate}

\item In \cite{extendQL}, only the Hilbert-style deduction system has been analyzed.

\item There is a mistake in the proof of the completeness theorem in \cite{extendQL} originating from symmetry frames.
Furthermore, in \cite{extendQL}, the proof of decidability of the validity problem of {\bf MB} is based on the finite model property, which is related to the proof of the completeness theorem. Therefore, 
it is important to reestablish decidability.

Here, an overview of the error is provided.
In proving the completeness theorem for a Hilbert-style deduction system
for modal logic with symmetry frames, the following problem arises.
To construct a finite canonical model for modal logic from an unprovable formula $A$, a set $\GA_A$ consisting of all subformulas of $A$ (and all their negative forms in some cases) is usually constructed.
In a canonical model, {\it consistent} subsets of $\GA_A$ are defined as possible worlds.
The binary relation $R$ of a canonical model is defined as follows:
$(\GA',\GA'') \in R$ if for all $\square B \in \GA'$, $B \in \GA''$.
To show symmetry, we must prove that $(\GA'',\GA') \in R$ also holds on this definition.
The following types of methods are generally used to prove this relation.
Suppose $\square B \in \GA''$. From $(\GA',\GA'') \in R$, $\square \neg \square B \notin \GA'$. 
Because $\neg B \to \square \neg \square B$ is provable, $\neg B \notin \GA'$.
Therefore, $B \in \GA'$.  
However, this proof fails as follows. Even if $\square B \in \GA_A$, there is no guarantee of $\square \neg \square B \in \GA_A$ because $\square \neg \square B$ is not a subformula of $\square B$.
This mistake is on page 562, line 12 of \cite{extendQL}.
This method works if an infinite set of all formulas, not just subformulas of $A$, is adopted as $\GA_A$. 
 (If completeness is all needed, we can change to this infinite model and use the method described in \cite{extendQL} to prove it.)
However, that method would make the canonical model infinite, and we could not prove the decidability.

\item {\bf MB} has only the modal comparison symbols.
Leaving room for analysis of the modal symbols corresponding to each number. (Details are provided in Section \ref{MLMB+}.)
\end{enumerate}

To solve these problems, in this study, 
{\it nested-sequent calculus} for {\bf MB} that satisfies the {\it cut-e\-li\-mi\-na\-tion theorem} is constructed, and the cut-free completeness theorem is proved.
The decidability of the validity problem of {\bf MB} is shown by using this new calculus.
In addition, a nested-sequent calculus for new logic {\bf MB+} ({\bf MB} with new modal symbols) is also constructed.

The concept of nested-sequent were introduced independently in \cite{Brünnler2006} \cite{Bull1992} \cite{kashima1994}  \cite{Poggiolesi2009}.
For logic that satisfies specific properties, using ordinary sequent may be inconvenient.
It is well known that in logics involving symmetry frames as semantics (e.g., {\bf S5} and {\bf B}), it is complex to construct the usual sequent calculus that satisfies the cut-elimination theorem.
Various developed sequent systems have been proposed to overcome this problem, including nested-sequent (also known as {\it tree-hypersequent}) and others such as {\it hypersequent}, and {\it labelled sequent}.
These developed sequents are structures constructed by combining multiple sequents.
In many cases, These developmental sequents contain semantic elements.
Intuitively, each sequent in nested-sequent or labelled-sequent corresponds to each possible world of a Kripke frame.
The nested-sequent have a tree-like structure with the sequents as nodes, which intuitively corresponds to the tree-like part of the Kripke frame.
One of the characteristics of tree-like sequents is that it is easy to translate the entire tree-like structure into a single formula by translating sequents into formulas, starting from the leaf sequents in turn.
A labelled-sequent uses specific labels to represent each possible world in the Kripke frame.
In these developed sequent calculi, when constructing a canonical model, transforming just one sequent ensures that the canonical model does not become an infinite model while preserving conditions such as symmetry.
In this study, we employ a nested-sequent, which exhibits relatively manageable properties among these candidates.
Studies about these developed sequents are discussed, for example, in \cite{AvronHypersequent} \cite{Gabbaylabel} \cite{negi2} \cite{negi} \cite{Poggiolesi2009} \cite{Poggiolesi2010}.
A comparison and summary of these developed sequents are discussed in \cite{Poggiolesi2023}.

In this study, we adopt a development of the usual nested-sequent.
In the nested-sequent of standard modal logic, brackets [\ ] represent modal concepts of $\square$. 
In other words, intuitively, [\ ] expresses the difference between possible worlds.
This part needs to be developed in nested-sequents for logics that use more complex notions of modality.
Because {\bf MB} includes the modal symbol $\square^{d}_{\alpha}$ to concretely express the number $\alpha$ of the absolute value of the inner product, in this study, we use the bracket [\ ]$^{d}_{\alpha}$. 
Except for this difference, almost the same concept as the standard nested-sequent is employed.

In section \ref{EQLaMB}, the basics of {\bf MB} are reviewed.
In section \ref{TreeB}, the basics of nested-sequent for {\bf MB} are defined.  
In section \ref{sNSMB}, a nested-sequent calculus for {\bf MB} is defined, and some theorems are established.
In section \ref{MLMB+}, a nested-sequent calculus for {\bf MB+} is 
discussed.

Because this study is entirely the result of mathematical logic, a more detailed explanation of the quantum mechanical background of {\bf MB} is omitted.
For such an explanation, see \cite{extendQL}.
For more detailed explanations of the quantum mechanical background of {\bf QL}, see \cite{bat1} \cite{chiara1} \cite{HANDQ} \cite{Chihilbert} \cite{chi:giu}.
For more details about recent studies of sequent calculi and developed sequent systems for {\bf QL}, see, for example, \cite{Fazio2023} \cite{TKlosi} \cite{TKsosi} \cite{Kornell2023} \cite{nis1}.

\section{Modal logic MB}\label{EQLaMB}
This section reviews {\bf MB} defined in \cite{extendQL}.
The language of {\bf MB} consists of the following vocabulary:

\begin{description}
\item propositional variables: $p, q, \ldots$
\item propositional
constants: $\top, \bot$
\item logical connectives: $\neg, \wedge, \square^{c}_{\alpha}, \square^{o}_{\alpha} \ (\alpha \in J)$ 
\end{description}

\noi
where $J$ is a finite subset of the unit interval $I = [0,1]$ that includes $0$ and $1$.
As in \cite{extendQL}, in this study, 
we assume that $J$ is fixed to one particular set.
$c$ stands for ``closed'', and $o$ stands for ``open''.
These meanings can be seen in the definition of the valuation of formulas in a frame, which will be discussed later.

The formulas of {\bf MB} are defined as follows:

\begin{description}
\item  $A :: =  p \ | \ \top \ | \ \bot \ | \ \neg A \ |\  A \AND A \ |\ \square^{c}_{\alpha} A \ |\ \square^{o}_{\alpha} A \ \ \ (\alpha \in J)$
\end{description}

Formulas are denoted $A, B, \ldots$,
and finite sets of formulas are denoted $\GA,\DE,\SI,\ldots$.
Elements of $\{c, o\}$ are denoted $d, d',\ldots$.
We use the following abbreviations.
$A \vee B = \neg (\neg A \wedge \neg B)$,
$A \to B = \neg A \vee B$,
$\Diamond^{c}_{\alpha} A = \neg \square^{c}_{\alpha} \neg A$,
$\Diamond^{o}_{\alpha} A = \neg \square^{o}_{\alpha} \neg A$.

An {\it EQL-frame} $(S, R)$ is defined as follows:

\begin{description}
\item $S$: a non-empty set, an element referred to as a possible
world (or physically, a pure quantum state).

\item $R$: an $I$-valued accessibility relation on $S$, i.e., $R$ : $S \times S \to I$,
satisfying the following conditions: $R(s, t) = 1$ iff $s = t$ (reflexivity),
$R(s, t) = R(t, s) (\forall s, t \in S)$ (symmetry). 
(This $R$ represents the absolute value of the inner product between states.)
\end{description}

We write $s(\alpha)t$ for $R(s,t) = \alpha$.

An {\it MB-realization} is a structure $M = (S, R, P, V)$, where
\begin{description}
\item $(S,R)$ is an EQL-frame.
\item $P$ is a set of subsets of $S$, including $S$ and $\emptyset$, being closed
under set-theoretic finite intersection, set-theoretic complement
relative to $S$, and the two series of operations $\square^{c}_{\alpha}$, $\square^{o}_{\alpha}$ on a set for each $\alpha \in J$
that are defined as follows:
\begin{description}
\item $\square^{c}_{\alpha} S' \stackrel{\mathrm{def}}{=} \{ s \in S | \forall t \in S \ (\alpha \leqq R(s, t)$ implies $t \in S'$)\}.
\item $\square^{o}_{\alpha} S' \stackrel{\mathrm{def}}{=} \{ s \in S | \forall t \in S \ (\alpha < R(s, t)$ implies $t \in S'$)\}.
\end{description}

(Although the modal symbols used here as operations on sets are the same as those in the language of {\bf MB}, these are defined independently of the language of {\bf MB}. This concept is introduced to ensure that when dealing with $V$, the sets of possible worlds are closed in $P$ in the operation of logical connective $\square^{d}_{\alpha}$ \cite{extendQL}.)
\item Valuation $V$ is a map from propositional variables to $P$.
\end{description}

$V$ is extended inductively as follows: 
\begin{description}
 \item $V( \top ) = S$,
 \item $V( \bot ) = \emptyset$,
\item $V(A \wedge B) = V(A) \cap V(B)$,
\item $V( \neg A ) = V(A)^{c}$,
\item $V( \square^{c}_{\alpha} A) = \{ s \in S |$ for all $t \in S$, if $\alpha \leqq R (s, t)$, then $t \in V(A) $ $\}$, 
\item $V( \square^{o}_{\alpha} A) = \{ s \in S |$ for all $t \in S$, if $\alpha < R (s, t)$, then $t \in V(A) $ $\}$.
\end{description}

Formula $A$ is {\it true} at $s \in S$ if $s \in V(A)$ and we write $s \models A$.
$A$ is {\it valid in an {\bf MB}-realization $(S, R, P, V)$} if for all $s \in S$, $A$ is true at $s$.
$A$ is {\it valid in an EQL-frame $(S, R)$} if for all $P$ and $V$, $A$ is valid in $(S, R, P, V)$. 
$A$ is {\it valid} if $A$ is valid in all EQL-frames.

\section{Nested-sequent}\label{TreeB}
This section defines the basics of the nested-sequent for {\bf MB}.

A {\it sequent} is a structure $\GA \Rightarrow \DE$, where $\GA$ and $\DE$ are 
finite sets of formulas. 
A {\it nested-sequent} is defined inductively as follows:

\begin{enumerate}
\item A sequent is a nested-sequent (a tree with only a root).
\item $\GA \Rightarrow \DE, \mathcal{T}$ is a nested-sequent where $\GA \Rightarrow \DE$ is a sequent and $\mathcal{T}$ is a finite set of nested-sequents enclosed in each {\it modal brackets} $[ \ ]^{d}_{\alpha}$ where $d \in \{ c,o\}$ and $\alpha \in J-\{1\}$.
\end{enumerate}

For example, $p \wedge r, q \Rightarrow q, [ \Rightarrow p, [\square^{c}_{0.3} r \Rightarrow p \wedge q]^{o}_{0.7} ]^{c}_{0.5}, [ r \Rightarrow p, q]^{o}_{0.1} $ is a nested-sequent.
A nested-sequent can be considered a {\it tree} structure if the leftmost sequent is regarded as the root, each internal sequent is considered a {\it node}, and each modal bracket is regarded as an {\it edge} labelled with $(\alpha, d)$.

\begin{figure}[htbp]
 \begin{center}
  \includegraphics[width=70mm]{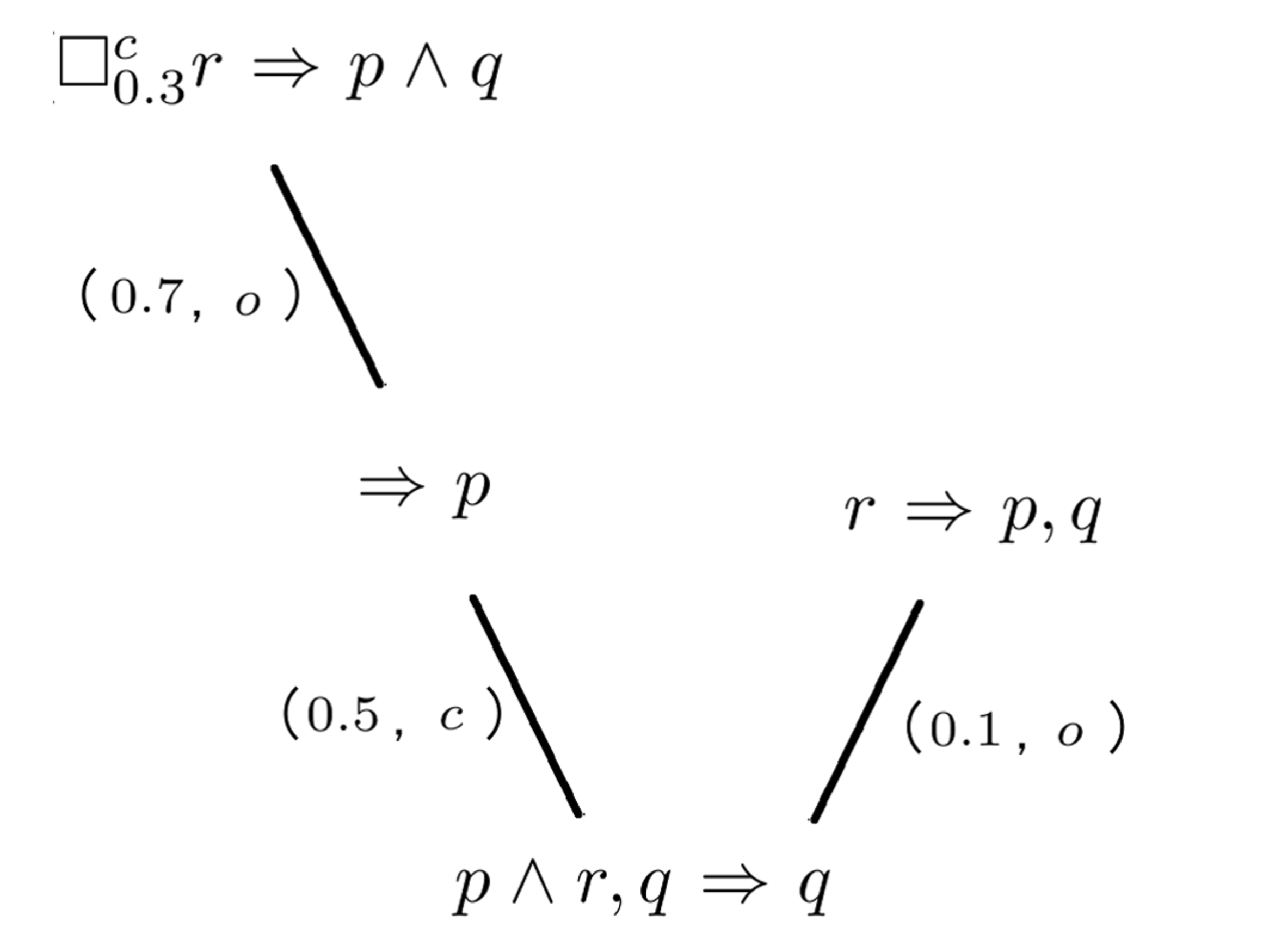}

 \label{fig:one}
Example: Tree representation of $p \wedge r, q \Rightarrow q, [ \Rightarrow p, [\square^{c}_{0.3} r \Rightarrow p \wedge q]^{o}_{0.7} ]^{c}_{0.5}, [ r \Rightarrow p, q]^{o}_{0.1} $.
\end{center}
\end{figure}

A number $\alpha$ {\it appears} in a nested-sequent $\GA \Rightarrow \DE, \mathcal{T}$ if $\square^{c}_{\alpha} A$ or $\square^{o}_{\alpha} A$ appear in it for some $A$, or some brackets $[\ ]^{d}_{\alpha}$ appear in it.
The set $(\GA \Rightarrow \DE, \mathcal{T})_N$ is defined as the set of all nodes of $\GA \Rightarrow \DE, \mathcal{T}$.
If the same sequent appears multiple times, they are treated as separate nodes.
For example, the first $p \Rightarrow q$ and the last $p \Rightarrow q$ in $p \Rightarrow q, [r \Rightarrow s, [p \Rightarrow q]^{o}_{0.7} ]^{c}_{0.5}$ are different nodes.
The ordered set $(\GA \Rightarrow \DE, \mathcal{T})_J$ is defined as the set of all $\alpha \in J$ that appear in $\GA \Rightarrow \DE, \mathcal{T}$ with $0$ and $1$. 
For example,
$(p \wedge r, q \Rightarrow q, [ \Rightarrow p, [\square^{c}_{0.2} r \Rightarrow p \wedge q]^{o}_{0.7} ]^{c}_{0.5}, [ r \Rightarrow p, q]^{o}_{0.1})_J = \{0, 0.1, 0.2, 0.5, 0.7, 1\}$.

We write $\| \GA \Rightarrow \DE, \mathcal{T} \|$ for the abbreviated nested-sequent in which $\GA \Rightarrow \DE, \mathcal{T}$ appears as a subtree.
This expression is used when focusing only on a specific part, $\GA \Rightarrow \DE, \mathcal{T}$, of a nested-sequent.
Note that even if $\GA \Rightarrow \DE, \mathcal{T}$ appears multiple times in a nested-sequent, when this notation is used, we are focusing on one particular subtree.
In a situation in which we focus on a specific $\GA \Rightarrow \DE, \mathcal{T}$ in a nested-sequent $\GA' \Rightarrow \DE', \mathcal{T'}$, we write  $\| \GA \Rightarrow \DE, \mathcal{T} \|$ =  $\GA' \Rightarrow \DE', \mathcal{T'}$.
After writing such an abbreviation, the discussion will proceed, assuming that the abbreviation is fixed.
For example, after writing $\| p \Rightarrow q \| =  p \Rightarrow q, [r \Rightarrow s]^{o}_{0.5}, [p \Rightarrow q]^{c}_{0.3}$ (and if it is determined from the context that $p \Rightarrow q$ refers to the first one), $\| p \Rightarrow q, r \| $ means $p \Rightarrow q, r, [r \Rightarrow s]^{o}_{0.5}, [p \Rightarrow q]^{c}_{0.3}$.

For convenience, in the following, we will equate the sequent $\GA \Rightarrow \DE$ with the nested-sequent $\GA \Rightarrow \DE, \emptyset$ that has the empty set of trees.
Therefore, if $\| \GA \Rightarrow \DE, \mathcal{T} \|$ is written, $ \GA \Rightarrow \DE$ may be a leaf of the tree.

The order $\prec$ on $I \times \{c,o\}$ is defined as follows:

\begin{description}

\item In case of $d = d'$ : \ $(\alpha, d) \prec (\beta, d')$ if $\alpha < \beta$.

\item In case of $d \neq d'$ : \ 
$(\alpha, c) \prec (\beta, o)$ if $\alpha \leqq \beta$.
$(\beta, o) \prec (\alpha, c)$ if $\alpha > \beta$.

\end{description}

Intuitively, this order represents the inverse of the inclusion relation of the upper closed subsets of $I$.
It is easy to see that this order is total.

We write $(\GA \Rightarrow \DE, \mathcal{T}) \triangleleft (\GA' \Rightarrow \DE', \mathcal{T'})$ if $\GA \Rightarrow \DE, \mathcal{T}$ is a subtree of $ \GA' \Rightarrow \DE', \mathcal{T'}$. In particular, if $\GA \Rightarrow \DE$ is a node of $\GA' \Rightarrow \DE', \mathcal{T'}$, we write $(\GA \Rightarrow \DE) \triangleleft (\GA' \Rightarrow \DE', \mathcal{T'})$.

An {\it embedding} of a nested-sequent $\GA \Rightarrow \DE, \mathcal{T}$ in an {\bf MB}-realization $(S,R,P,V)$ is a function $\mathcal{E}$ from $(\GA \Rightarrow \DE, \mathcal{T})_N$ to $S$ that satisfies the following conditions:

\begin{description}

\item If $(\GA_1 \Rightarrow \DE_1, [\GA_2 \Rightarrow \DE_2, \mathcal{T'}]^{c}_{\alpha}) \triangleleft (\GA \Rightarrow \DE, \mathcal{T})$ and $R((\mathcal{E}(\GA_1 \Rightarrow \DE_1),(\mathcal{E}(\GA_2 \Rightarrow \DE_2)) = \beta$, then $ \alpha \leqq \beta$.

\item If $(\GA_1 \Rightarrow \DE_1, [\GA_2 \Rightarrow \DE_2, \mathcal{T'}]^{o}_{\alpha}) \triangleleft (\GA \Rightarrow \DE, \mathcal{T})$ and $R((\mathcal{E}(\GA_1 \Rightarrow \DE_1),(\mathcal{E}(\GA_2 \Rightarrow \DE_2)) = \beta$, then $ \alpha < \beta$.

\end{description}

A nested-sequent $\GA \Rightarrow \DE, \mathcal{T}$ is {\it false} in an {\bf MB}-realization $(S,R,P,V)$ under $\mathcal{E}$ 
if for
all sequents $\GA' \Rightarrow \DE'$ in $\GA \Rightarrow \DE, \mathcal{T}$, all $A \in \GA'$ are true at $\mathcal{E}(\GA' \Rightarrow \DE')$ and all $A \in \DE'$ are false at $\mathcal{E}(\GA' \Rightarrow \DE')$.
A nested-sequent $\GA \Rightarrow \DE, \mathcal{T}$ is {\it true} in $(S,R,P,V)$ under $\mathcal{E}$ if $\GA \Rightarrow \DE, \mathcal{T}$ is not false in $(S,R,P,V)$ under $\mathcal{E}$.
A nested-sequent $\GA \Rightarrow \DE, \mathcal{T}$ is {\it valid in} $(S,R,P,V)$ if for all $\mathcal{E}$, $\GA \Rightarrow \DE, \mathcal{T}$ is true under $\mathcal{E}$.
A nested-sequent $\GA \Rightarrow \DE, \mathcal{T}$ is {\it valid} 
if it is valid in all $(S,R,P,V)$.

The {\it interpretation} $\tau$ of a nested-sequent to a formula is defined inductively as follows:

\begin{description}

\item $\tau (\GA \Rightarrow \DE) = \bigwedge \GA \to \bigvee \DE$.

\item $\tau (\GA \Rightarrow \DE, [\GA_1 \Rightarrow \DE_1, \mathcal{T}_1]^{d_1}_{\alpha_1}, \ldots,  [\GA_n \Rightarrow \DE_n, \mathcal{T}_n]^{d_n}_{\alpha_n} )$

$= \tau (\GA \to \DE) \vee  \square^{d_1}_{\alpha_1} \tau (\GA_1 \Rightarrow \DE_1, \mathcal{T}_1) \vee \ldots \vee \square^{d_n}_{\alpha_n} \tau (\GA_n \Rightarrow \DE_n, \mathcal{T}_n) $.

\end{description}

\noi
where $\bigwedge \GA$ denotes a formula connecting all the formulas in $\GA$ with $\wedge$, and $\bigvee \DE$ denotes a formula connecting all the formulas in $\DE$ with $\vee$.

As in the case of other studies of nested-sequent, the following theorem holds. 
\begin{theorem}

$\GA \Rightarrow \DE, \mathcal{T}$ is valid iff $\tau(\GA \Rightarrow \DE, \mathcal{T})$ is valid.

\end{theorem}
\begin{proof}
 $\tau(\GA \Rightarrow \DE, \mathcal{T})$ generally has the following form:

\begin{description}
 \item $(\bigwedge \GA \to \bigvee \DE) \vee  \square^{d_1}_{\alpha_1} ((\bigwedge \GA_1 \to \bigvee \DE_1) \vee T^1_1 \vee \ldots \vee T^1_m) \vee \ldots \vee \square^{d_n}_{\alpha_n} ((\bigwedge \GA_n \to \bigvee \DE_n) \vee T^n_1 \vee \ldots \vee T^n_l)$.
\end{description}

Suppose $\GA \Rightarrow \DE, \mathcal{T}$ is false under $\mathcal{E}$. 
Then, $\bigwedge \GA \to \bigvee \DE$ is false at $\mathcal{E}(\GA \Rightarrow \DE)$.
Furthermore, for all $i \in \{ 1,\ldots,n \}$, $ \GA_i \Rightarrow \DE_i$ is false at $\mathcal{E}(\GA_i \Rightarrow \DE_i)$ and $\alpha_i \leqq R(\mathcal{E}(\GA \Rightarrow \DE),\mathcal{E}(\GA_i \Rightarrow \DE_i))$ (if $d_i = c$) or $\alpha_i < R(\mathcal{E}(\GA \Rightarrow \DE),\mathcal{E}(\GA_i \Rightarrow \DE_i))$ (if $d_i = o$).
Continuing this procedure up to all leaves of the tree confirms that for all $i \in \{ 1,\ldots,n \}$ and for each $j$, $\square^{d_i}_{\alpha_i} ((\bigwedge \GA_i \to \bigvee \DE_i) \vee T^i_1 \vee \ldots \vee T^i_j)$ is false at $\mathcal{E}(\GA \Rightarrow \DE)$. 
Then, $\tau(\GA \Rightarrow \DE, \mathcal{T})$ is false at $\mathcal{E}(\GA \Rightarrow \DE)$.

Suppose $\tau(\GA \Rightarrow \DE, \mathcal{T})$ is false at $x \in S$.
Then $\GA \Rightarrow \DE$ is false at $x$.
Furthermore, for all $i \in \{ 1,\ldots,n \} $, there exists $x_i \in S$ such that $ \GA_i \Rightarrow \DE_i$ is false at $x_i$ and
$\alpha_i \leqq R(x,x_i)$ (if $d_i = c$) or $\alpha_i < R(x,x_i)$ (if $d_i = o$). This notion applies inductively to each $T^i_j$ until it reaches the leaves.
$\mathcal{E}$ is defined as a function that transfers each sequent to each element that makes it false.
That is, $\mathcal{E}(\GA \Rightarrow \DE) = x$, 
$\mathcal{E}(\GA_1 \Rightarrow \DE_1) = x_1, \ldots$.
Then, $\GA \Rightarrow \DE, \mathcal{T}$ is false under $\mathcal{E}$.

\end{proof}

\section{Nested-sequent calculus {\bf NSMB}}\label{sNSMB}

This section discusses the nested-sequent calculus for {\bf MB} that satisfies the cut-elimination theorem.
The nested-sequent calculus {\bf NSMB} is defined as follows:



\noi
Axioms$^{}$: \begin{center} $\SEQ{\| A}{A, \mathcal{T} \|}$ \ \ \ \ $\SEQ{ \|}{ \top, \mathcal{T}\|}$ \ \ \ \  $\SEQ{ \| \bot}{ \mathcal{T}\|}$  \ \ \ \  $\SEQ{ \| }{\square^{o}_{1} A, \mathcal{T}\|}$  \end{center} \qquad 
\\
Rules:
$$
\infer[\mbox{(cut)}]{\SEQ{\| \GA}{\DE, \mathcal{T} \|}}{
 \SEQ{\| \GA}{\DE, A, \mathcal{T} \|} & \SEQ{\| A, \GA}{\DE, \mathcal{T} \|}}
$$ 
$$
\infer[\mbox{(wL)$^{}$}]{\SEQ{\| A, \GA}{\DE, \mathcal{T} \|}}
 {\SEQ{\| \GA}{\DE, \mathcal{T} \|}}
\hspace{5mm}
\infer[\mbox{(wR)$^{}$}]{\SEQ{\| \GA}{\DE, A, \mathcal{T} \|}}
 {\SEQ{\| \GA}{\DE, \mathcal{T} \|}}
\hspace{5mm}
\infer[\mbox{(\NOT L)$^{}$}]{\SEQ{\| \neg A, \GA}{\DE, \mathcal{T} \|}}
 {\SEQ{\| \GA}{\DE, A, \mathcal{T} \|}}
\hspace{5mm}
\infer[\mbox{(\NOT R)$^{}$}]{\SEQ{\| \GA}{\DE, \neg A, \mathcal{T} \|}}
 {\SEQ{\| A, \GA}{\DE, \mathcal{T} \|}}
$$
$$
\infer[\mbox{(\AND L)$^{}$}]{\SEQ{\| A \wedge B, \GA}{\DE, \mathcal{T} \|}}
 {\SEQ{\| A, B, \GA}{\DE, \mathcal{T} \|}}
\hspace{5mm}
\infer[\mbox{(\AND R)}]{\SEQ{\| \GA}{\DE, A \wedge B, \mathcal{T} \|}}{
 \SEQ{\| \GA}{\DE, A, \mathcal{T}\|} & \SEQ{\| \GA}{\DE, B, \mathcal{T}\|}}
$$ 
$$
\infer[\mbox{($\square$ L) $^{(1)}$}]{\SEQ{\| \square^{d}_{\alpha} A, \GA}{\DE, \ [\GA' \Rightarrow \DE', \mathcal{T'}]^{d'}_{\beta}, \mathcal{T} \|}}{
 \SEQ{\| \GA}{\DE, [A, \GA' \Rightarrow \DE', \mathcal{T'}]^{d'}_{\beta}, \mathcal{T} \|}}
\hspace{5mm}
\infer[\mbox{($\square$ L sym) $^{(1)}$}]{\SEQ{\| \GA}{\DE, \ [\square^{d}_{\alpha} A, \GA' \Rightarrow \DE', \mathcal{T'}]^{d'}_{\beta}, \mathcal{T} \|}}{
 \SEQ{\| A, \GA}{\DE, [\GA' \Rightarrow \DE', \mathcal{T'}]^{d'}_{\beta}, \mathcal{T} \|}}
$$
$$
\infer[\mbox{($\square$ L self) $^{(2)}$}]{\SEQ{\| \square^{d}_{\alpha} A, \GA}{\DE, \mathcal{T} \|}}{
 \SEQ{\| A, \GA}{\DE, \mathcal{T} \|}}
\hspace{5mm}
\infer[\mbox{($\square^c_0)^{(3)}$ }]{\SEQ{\| \square^c_0 A, \GA'}{\DE', \mathcal{T'} \|}}{
 \SEQ{\| A, \GA}{\DE, \mathcal{T} \|}}
\hspace{5mm}
$$
$$
\infer[\mbox{($\square$ R) $^{}$}]{\SEQ{\| \GA}{\DE, \square^{d}_{\alpha} A , \mathcal{T} \|}}{
 \SEQ{\| \GA}{\DE, [\Rightarrow A]^{d}_{\alpha}, \mathcal{T} \|}}
\hspace{5mm}
\infer[\mbox{($\square$ R self) }]{\SEQ{\| \GA}{\DE, \square^c_1 A, \mathcal{T} \|}}{
 \SEQ{\| \GA}{\DE, A, \mathcal{T} \|}}
\hspace{5mm}
$$

\begin{description}

\item * In all rules except $(\square^c_0)$, the parts other than those specified parts must be the same at the top and bottom.
For example, in ($\neg$ L), the only difference between the upper and lower nested-sequents is the change from
$\GA \Rightarrow \DE, A$ to $\neg A, \GA \Rightarrow \DE$ in the stated node $\GA \Rightarrow \DE, A$.
In the case of (cut) and ($\wedge$R), this condition is also imposed on the top two sequents. In the case of (cut) and ($\wedge$R), the top two and the bottom one nested-sequents must be the same for all three except for the stated parts.

\item (1) $(\alpha, d) \preceq (\beta, d')$ . 

\item (2) $(\alpha, d) \neq (1, o)$. 

\item (3) This rule erases $A$ from the left of one node in the tree and adds $\square^c_0 A$ to the left of another arbitrary node of the same tree.
\end{description}

The following deduction is an example of a proof of $A \Rightarrow \square^{c}_{0.5}  \Diamond^{o}_{0.3} A$ in {\bf NSMB}.

$$
\infer[\mbox{($\square$ R)}]{\SEQ{ A}{\square^{c}_{0.5}  \Diamond^{o}_{0.3} A}}
{
\infer[\mbox{($\neg$ R)}]{\SEQ{ A}{ [\Rightarrow \neg \square^{o}_{0.3} \neg A ]^{c}_{0.5}}}
{
\infer[\mbox{($\square$ L sym)}]{\SEQ{ A}{[\square^{o}_{0.3} \neg A \Rightarrow ]^{c}_{0.5}}}
{
    \infer[\mbox{($\neg$ L)}]{\SEQ{\neg A, A}{[\Rightarrow ]^{c}_{0.5}}}
    {
     {\SEQ{A}{A, [\Rightarrow ]^{c}_{0.5}}}
    }
  }
}
}
$$

\begin{theorem}[Soundness theorem for {\bf NSMB}]\label{soTSMB} If $\GA \Rightarrow \DE, \mathcal{T}$ is provable in {\bf NSMB}, then
$\GA \Rightarrow \DE, \mathcal{T}$ is valid.
\end{theorem}
\begin{proof}
It is proved by induction on the construction of the proof of nested-sequent $\GA \Rightarrow \DE, \mathcal{T}$.
We only show the cases in which the last rule used in the proof is ($\square$ L) or ($\square$ R).
The proofs for the other cases are simpler. 
First, we show the case in which the
last rule is ($\square$ L).
 
Suppose that $\| \square^{d}_{\alpha} A, \GA \Rightarrow \DE, \ [\GA' \Rightarrow \DE',\mathcal{T'}]^{d'}_{\beta}, \mathcal{T}\|$
is false in $(S,R,P,V)$ under embedding $\mathcal{E}$.
Then, $\square^{d}_{\alpha} A$ is true at $\mathcal{E}(\square^{d}_{\alpha} A, \GA \Rightarrow \DE)$.
From the condition of the rule,
$(\alpha, d) \preceq (\beta, d')$. 
\begin{description}

\item In the case of $d=d'=c$, from the definition of embedding, 
$\beta \leqq R((\mathcal{E}(\square^{c}_{\alpha} A, \GA \Rightarrow \DE), (\mathcal{E}(\GA' \Rightarrow \DE'))$. Therefore, $\alpha \leqq R((\mathcal{E}(\square^{c}_{\alpha} A, \GA \Rightarrow \DE), (\mathcal{E}(\GA' \Rightarrow \DE'))$.

\item In the case of $d=c$ and $d'=o$, from the definition of embedding, 
$\beta < R((\mathcal{E}(\square^{c}_{\alpha} A, \GA \Rightarrow \DE), (\mathcal{E}(\GA' \Rightarrow \DE'))$. Therefore, $\alpha < R((\mathcal{E}(\square^{c}_{\alpha} A, \GA \Rightarrow \DE), (\mathcal{E}(\GA' \Rightarrow \DE'))$.

\item In the case of $d=o$ and $d'=c$, from the definition of
$\prec$, $\alpha < \beta$.
From the definition of embedding, 
$\beta \leqq R((\mathcal{E}(\square^{o}_{\alpha} A, \GA \Rightarrow \DE), (\mathcal{E}(\GA' \Rightarrow \DE'))$. Therefore, $\alpha < R((\mathcal{E}(\square^{o}_{\alpha} A, \GA \Rightarrow \DE), (\mathcal{E}(\GA' \Rightarrow \DE'))$.

\item In the case of $d=d'=o$, from the definition of
$\prec$, $\alpha < \beta$.
From the definition of embedding, 
$\beta < R((\mathcal{E}(\square^{o}_{\alpha} A, \GA \Rightarrow \DE), (\mathcal{E}(\GA' \Rightarrow \DE'))$. Therefore, $\alpha < R((\mathcal{E}(\square^{o}_{\alpha} A, \GA \Rightarrow \DE), (\mathcal{E}(\GA' \Rightarrow \DE'))$.

\end{description}

Therefore, in any case, $A$ is true at $\mathcal{E}(\GA' \Rightarrow \DE')$, 
and $\| \GA \Rightarrow \DE, \ [A, \GA' \Rightarrow \DE',\mathcal{T'}]^{d'}_{\beta}, \mathcal{T}\|$ is 
false under $\mathcal{E'}$ where $\mathcal{E'}$ is exactly the same as $\mathcal{E}$ except that $\mathcal{E'}(\GA \Rightarrow \DE) = \mathcal{E}(\square^{d}_{\alpha} A, \GA \Rightarrow \DE)$ and $\mathcal{E'}(A, \GA' \Rightarrow \DE') = \mathcal{E}(\GA' \Rightarrow \DE')$.

Next, we show the case where the
last rule is ($\square$ R).  
Suppose that $\| \GA \Rightarrow \DE, \square^{d}_{\alpha} A, \mathcal{T} \|$ is false in $(S,R,P,V)$ under $\mathcal{E}$.
Then there exists $s \in S$ such that $\mathcal{E}(\GA \Rightarrow \DE, \square^{d}_{\alpha} A ) (\beta) s$, $\alpha \leqq \beta$ (if $d =c$), $\alpha < \beta$ (if $d = o$), and $A$ is false at $s$.
Let $\mathcal{E'}$ be the embedding from $\| \GA \Rightarrow \DE, [ \Rightarrow A ]^{d}_{\alpha} \|$ to $(S,R,P,V)$ such that $\mathcal{E'}(\Rightarrow A) = s$,
$\mathcal{E'}(\GA \Rightarrow \DE) = \mathcal{E}(\GA \Rightarrow \DE, \square^{d}_{\alpha} A) $ and $\mathcal{E'} = \mathcal{E}$ for the other sequents.
\ Then, $\| \GA \Rightarrow \DE, [ \Rightarrow A ]^{d}_{\alpha} \|$
is false in $(S,R,P,V)$ under $\mathcal{E}'$.
\end{proof}

For the completeness theorem, the contraposition of the theorem is proved. In
other words, we show that if a nested-sequent $\GA \Rightarrow \DE, \mathcal{T}$ is not provable in {\bf NSMB}, then an {\bf MB}-realization $(S,R,P,V)$ exists with an embedding $\mathcal{E}$ of $\GA \Rightarrow \DE, \mathcal{T}$ to $(S,R,P,V)$ such that $\GA \Rightarrow \DE, \mathcal{T}$ is
false in $(S,R,P,V)$ under $\mathcal{E}$.

Suppose $\GA \Rightarrow \DE, \mathcal{T}$ is not provable.
(We assume that $\GA \Rightarrow \DE, \mathcal{T}$ is fixed to one particular nested-sequent to the end of this section.)
To construct a model in which $\GA \Rightarrow \DE, \mathcal{T}$ is false, a new nested-sequent $\GA_C \Rightarrow \DE_C, \mathcal{T}_C$ is formed from $\GA \Rightarrow \DE, \mathcal{T}$ by the following iterative procedure.
This procedure is continued until the nested-sequent is no longer changed by applying any of the following steps.
Changes in the sequent are denoted by $\GA_{0} \Rightarrow \DE_{0}, \mathcal{T}_0  (=\GA \Rightarrow \DE, \mathcal{T}), 
\GA_{1} \Rightarrow \DE_{1}, \mathcal{T}_{1},...,
\GA_{i} \Rightarrow \DE_{i}, \mathcal{T}_{i}, 
\GA_{i+1} \Rightarrow \DE_{i+1}, \mathcal{T}_{i+1}....$

\begin{enumerate}

 \item\label{wedgeL}

If $\| \GA' \Rightarrow \DE', \mathcal{T}' \| = \GA_{i} \Rightarrow \DE_{i}, \mathcal{T}_{i}$ and $A \wedge B \in \GA'$, then we construct $\GA_{i+1} \Rightarrow \DE_{i+1}, \mathcal{T}_{i+1}$ by
adding $A$ and $B$ to $\GA'$ of $\GA_{i} \Rightarrow \DE_{i}, \mathcal{T}_{i}$.
That is, $\GA_{i+1} \Rightarrow \DE_{i+1}, \mathcal{T}_{i+1}$ = $\| A,B,\GA' \Rightarrow \DE', \mathcal{T}' \|$.
This new nested-sequent is also not provable because of the rule ($\wedge$L).

\item\label{wedgeR}
If $\| \GA' \Rightarrow \DE', \mathcal{T}' \| = \GA_{i} \Rightarrow \DE_{i}, \mathcal{T}_{i}$ and $A \wedge B \in \DE'$,
at least one of $\| \GA' \Rightarrow \DE', A, \mathcal{T}' \|$ and $\| \GA' \Rightarrow \DE' , B, \mathcal{T}' \|$ is not provable because of the rule ($\wedge$R).
Of these, the unprovable one is adopted as $\GA_{i+1} \Rightarrow \DE_{i+1}, \mathcal{T}_{i+1}$.

\item\label{negL}
If $\| \GA' \Rightarrow \DE', \mathcal{T}' \| = \GA_{i} \Rightarrow \DE_{i}, \mathcal{T}_{i}$ and $\neg A \in \GA'$, then we construct $\GA_{i+1} \Rightarrow \DE_{i+1}, \mathcal{T}_{i+1} = \| \GA' \Rightarrow \DE', A, \mathcal{T}' \|$.  
This new nested-sequent is also not provable because of the rule ($\neg$L).

\item\label{negR}
If $\| \GA' \Rightarrow \DE', \mathcal{T}' \| = \GA_{i} \Rightarrow \DE_{i}, \mathcal{T}_{i}$ and $\neg A \in \DE'$, then we construct $\GA_{i+1} \Rightarrow \DE_{i+1}, \mathcal{T}_{i+1} = \| A, \GA' \Rightarrow \DE', \mathcal{T}' \|$.  
This new nested-sequent is also not provable because of the rule ($\neg$R).

\item\label{boxL}
If $\| \GA' \Rightarrow \DE', [\GA'' \Rightarrow \DE'', \mathcal{T}'']^{d'}_{\beta}, \mathcal{T}' \| = \GA_{i} \Rightarrow \DE_{i}, \mathcal{T}_{i}$, $(\alpha, d) \preceq (\beta, d')$, and $\square^{d}_{\alpha} A \in \GA'$, then we construct $\GA_{i+1} \Rightarrow \DE_{i+1}, \mathcal{T}_{i+1} = \| \GA' \Rightarrow \DE', [A, \GA'' \Rightarrow \DE'', \mathcal{T}'']^{d'}_{\beta}, \mathcal{T}' \|$.  
This new nested-sequent is also not provable because of the rule ($\square$L).

\item\label{boxLsym}
If $\| \GA' \Rightarrow \DE', [\GA'' \Rightarrow \DE'', \mathcal{T}'']^{d'}_{\beta}, \mathcal{T}' \| = \GA_{i} \Rightarrow \DE_{i}, \mathcal{T}_{i}$, $(\alpha, d) \preceq (\beta, d')$, and $\square^{d}_{\alpha} A \in \GA''$, then we construct $\GA_{i+1} \Rightarrow \DE_{i+1}, \mathcal{T}_{i+1} = \| A, \GA' \Rightarrow \DE', [\GA'' \Rightarrow \DE'', \mathcal{T}'']^{d'}_{\beta}, \mathcal{T}' \|$.  
This new nested-sequent is also not provable because of the rule ($\square$L sym).

\item\label{boxLself}
If $\| \GA' \Rightarrow \DE', \mathcal{T}' \| = \GA_{i} \Rightarrow \DE_{i}, \mathcal{T}_{i}$, and $\square^{d}_{\alpha} A \in \GA'$ $((\alpha, d) \neq (1, o))$, then we construct $\GA_{i+1} \Rightarrow \DE_{i+1}, \mathcal{T}_{i+1} = \| A, \GA' \Rightarrow \DE', \mathcal{T}' \|$.  
This new nested-sequent is also not provable because of the rule ($\square$L self).

\item\label{boxR}
If $\| \GA' \Rightarrow \DE', \mathcal{T}' \| = \GA_{i} \Rightarrow \DE_{i}, \mathcal{T}_{i}$, and $\square^{d}_{\alpha} A \in \DE'$ $(\alpha \neq 1)$, then we construct $\GA_{i+1} \Rightarrow \DE_{i+1}, \mathcal{T}_{i+1} = \| \GA' \Rightarrow \DE', [\Rightarrow A ]^{d}_{\alpha}, \mathcal{T}' \|$.  
This new nested-sequent is also not provable because of the rule ($\square$R).
This step is performed once per occurrence of $\square^{d}_{\alpha} A$.

\item\label{boxRself}
If $\| \GA' \Rightarrow \DE', \mathcal{T}' \| = \GA_{i} \Rightarrow \DE_{i}, \mathcal{T}_{i}$, and $\square^{c}_{1} A \in \DE'$, then we construct $\GA_{i+1} \Rightarrow \DE_{i+1}, \mathcal{T}_{i+1} = \| \GA' \Rightarrow \DE', A, \mathcal{T}' \|$.  
This new nested-sequent is also not provable because of the rule ($\square$R self).

\item\label{boxC0}
If $\| \GA' \Rightarrow \DE', \mathcal{T}' \| = \| \GA'' \Rightarrow \DE'', \mathcal{T}'' \| = \GA_{i} \Rightarrow \DE_{i}, \mathcal{T}_{i}$, that is, $\GA' \Rightarrow \DE'$ and $\GA'' \Rightarrow \DE''$ are (could be the same) nodes of $\GA_{i} \Rightarrow \DE_{i}, \mathcal{T}_{i}$, and if $\square^{c}_{0} A \in \GA'$, then we construct $\GA_{i+1} \Rightarrow \DE_{i+1}, \mathcal{T}_{i+1} = \| A, \GA'' \Rightarrow \DE'' \mathcal{T}'' \|$.  
This new nested-sequent is also not provable because of the rule ($\square^{c}_{0}$).

\end{enumerate}

This procedure stops within a finite number of steps for the following reasons:
\begin{itemize}
\item The number of nodes and formulas appearing in $\GA_{i} \Rightarrow \DE_{i}, \mathcal{T}_{i}$ is always finite.
\item All of the procedures decrease the complexity of the formulas.
\item Step \ref{boxR} increases the number of nodes, but it is applied only once at most for one formula. 
In this procedure, only subformulas of the formulas in the first nested-sequent appear.
Therefore, the number of nodes can only increase by a finite amount from the initial nested-sequent.
\end{itemize}

Let $\GA_C \Rightarrow \DE_C, \mathcal{T}_C$ be the nested-sequent obtained at the end of this procedure, that is not provable. 
A {\it canonical model} is constructed from $\GA_C \Rightarrow \DE_C, \mathcal{T}_C$ with the following notion. 

We say a ordered set $\mathbb{U}$ is an {\it interpolated set} of $(\GA \Rightarrow \DE, \mathcal{T})_J$
if it satisfies the following conditions:

\begin{enumerate}
\item $(\GA \Rightarrow \DE, \mathcal{T})_J \subset \mathbb{U}$

 \item If $\alpha \in (\GA \Rightarrow \DE, \mathcal{T})_J$, $\beta \in (\GA \Rightarrow \DE, \mathcal{T})_J$, $\alpha \neq \beta$, and there is no $\gamma \in (\GA \Rightarrow \DE, \mathcal{T})_J$ that satisfies $\alpha < \gamma < \beta$, then there exists exactly one $\delta \in I$ in $\mathbb{U}$ that satisfies $\alpha < \delta < \beta$.  
\end{enumerate}

For example, $\{0, 0.05, 0.1, 0.15, 0.2, 0.4, 0.7, 0.9, 1\}$ is an interpolated set of $\{0, 0.1, 0.2, 0.7, 1\}$.
This set is necessary to ensure that all modalities do not affect each other when constructing a canonical model.
We write $Suc(\alpha)$ for the successor of element $\alpha$ in an interpolated set with $Suc(1) = 1$.

Let $\mathbb{U}_C$ be a certain interpolated set of $(\GA_C \Rightarrow \DE_C, \mathcal{T}_C)_J$.
A canonical model $(S_C,R_C, P_C, V_C)$ of $\GA_C \Rightarrow \DE_C, \mathcal{T}_C$ (with $\mathbb{U}_C$) is defined as follows:

\begin{description}

\item $S_C \stackrel{\mathrm{def}}{=} (\GA_C \Rightarrow \DE_C, \mathcal{T}_C)_N$

\item $R_C$: Defined in the following cases:
\renewcommand{\labelenumi}{(\Roman{enumi})}\begin{enumerate}
\item If $\GA' \Rightarrow \DE', [\GA'' \Rightarrow \DE'', \mathcal{T}'']^{c}_{\beta}, \mathcal{T}' \triangleleft \GA_C \Rightarrow \DE_C, \mathcal{T}_C$, then $R_C((\GA' \Rightarrow \DE'), (\GA'' \Rightarrow \DE'')) \stackrel{\mathrm{def}}{=} \beta$.

\item If $\GA' \Rightarrow \DE', [\GA'' \Rightarrow \DE'', \mathcal{T}'']^{o}_{\beta}, \mathcal{T}' \triangleleft \GA_C \Rightarrow \DE_C, \mathcal{T}_C$, then $R_C((\GA' \Rightarrow \DE'), (\GA'' \Rightarrow \DE'')) \stackrel{\mathrm{def}}{=} Suc(\beta)$.

\item $R_C((\GA' \Rightarrow \DE'), (\GA' \Rightarrow \DE')) \stackrel{\mathrm{def}}{=} 1$. (Same nodes)

\item In all other cases, $R_C((\GA' \Rightarrow \DE'), (\GA'' \Rightarrow \DE'')) \stackrel{\mathrm{def}}{=} 0$.

\end{enumerate}

\item $P_C$ $\stackrel{\mathrm{def}}{=} \{ S' \subseteq S|\exists A\  V_C(A) = S'  \}$

\item $V_C(p) \stackrel{\mathrm{def}}{=} \{ \GA' \Rightarrow \DE' | p \in \GA'\}$

\end{description}

\begin{lemma}\label{TSMBreal} 
$(S_C,R_C, P_C, V_C)$ is an {\bf MB}-realization.
\end{lemma}
\begin{proof}

By the definition of $R_C$, every pair of nodes is associated with a single number.
Furthermore, it is only in the case of $s = t$ that $R_C(s,t) =1$
for the following reasons.
From the definition of the bracket in a nested-sequent,
if $\GA' \Rightarrow \DE', [\GA'' \Rightarrow \DE'', \mathcal{T}'']^{c}_{\beta}, \mathcal{T}' \triangleleft \GA_C \Rightarrow \DE_C, \mathcal{T}_C$, then $\beta \neq 1$,
and if $\GA' \Rightarrow \DE', [\GA'' \Rightarrow \DE'', \mathcal{T}'']^{o}_{\beta}, \mathcal{T}' \triangleleft \GA_C \Rightarrow \DE_C, \mathcal{T}_C$, then $Suc(\beta) \neq 1$,
because of the definition of $\mathbb{U}_C$ and $\beta \in (\GA_C \Rightarrow \DE_C, \mathcal{T}_C)_J$.

The definition of $V$ for compound formulas corresponds to each condition of $P$.
For example, $V(A \wedge B) = V(A) \cap V(B)$ corresponds to the condition that $P$ is closed under a set-theoretic finite intersection.
Therefore, $P_C$ meets the conditions of $P$.
\end{proof}

The embedding $\mathcal{E}_C$ form $\GA \Rightarrow \DE, \mathcal{T}$
to $(S_C,R_C, P_C, V_C)$ is defined as follows.
From the configuration of $\GA_C \Rightarrow \DE_C, \mathcal{T}_C$, all the nodes that existed in $\GA \Rightarrow \DE, \mathcal{T} (=  \GA_{0} \Rightarrow \DE_{0}, \mathcal{T}_0)$ also exist in $\GA_C \Rightarrow \DE_C, \mathcal{T}_C$ (but with the added formulas).
$\mathcal{E}_C$ is defined as a function that transfers to that ``same'' node.
It can be proved from the composition of $\GA_C \Rightarrow \DE_C, \mathcal{T}_C$ and the definition of $R_C$ that $\mathcal{E}_C$ satisfies the embedding conditions.

\begin{lemma}\label{incS}
If $\mathcal{E}_C(\GA' \Rightarrow \DE') = \GA'' \Rightarrow \DE''$ and $A \in \GA' (A \in \DE')$, then $A \in \GA'' (A \in \DE'')$.
\end{lemma}
\begin{proof}
All steps do not remove formulas in the composition of $\GA_C \Rightarrow \DE_C, \mathcal{T}_C$. Therefore, all formulas present in $\GA_{0} \Rightarrow \DE_{0}, \mathcal{T}_0$ remain in  $\GA_C \Rightarrow \DE_C, \mathcal{T}_C$.
\end{proof}

\begin{lemma}\label{fuTSMB} 
For all $(\GA' \Rightarrow \DE') \in S_C$, if $A \in \GA'$, then $A$ is true at $ \GA' \Rightarrow \DE' \in S_C$.
 If $A \in \DE'$, then $A$ is false at $ \GA' \Rightarrow \DE' \in S_C$.
\end{lemma}
\begin{proof}
It is proved by induction on the construction of the formulas in $\GA'$ and $\DE'$.

\begin{itemize}
\item From the definition of $V_C$, the axiom $\| A \Rightarrow A, \mathcal{T}\|$, and the unprovability of $\GA_C \Rightarrow \DE_C, \mathcal{T}_C$,
$(\GA' \Rightarrow \DE') \models p$ if $p \in \GA'$ and $(\GA' \Rightarrow \DE') \not\models p$ if $p \in \DE'$.

\item Suppose $A \wedge B \in \GA'$. From Step \ref{wedgeL}, $A \in \GA'$ and $B \in \GA'$.
From the inductive hypothesis, $(\GA' \Rightarrow \DE') \models A$ and $(\GA' \Rightarrow \DE') \models B$.
Therefore, $(\GA' \Rightarrow \DE') \models A \wedge B$.

\item Suppose $A \wedge B \in \DE'$. From Step \ref{wedgeR},
at least one of $A \in \DE'$ or $B \in \DE'$ is established.
From the inductive hypothesis, $(\GA' \Rightarrow \DE') \not\models A$ or $(\GA' \Rightarrow \DE') \not\models B$.
Therefore, $(\GA' \Rightarrow \DE') \not\models A \wedge B$.

\item Suppose $\neg A \in \GA'$. From Step \ref{negL}, $A \in \DE'$. 
From the inductive hypothesis, $(\GA' \Rightarrow \DE') \not\models A$.
Therefore, $(\GA' \Rightarrow \DE') \models \neg A$.

\item Suppose $\neg A \in \DE'$. From Step \ref{negR}, $A \in \GA'$. 
From the inductive hypothesis, $(\GA' \Rightarrow \DE') \models A$.
Therefore, $(\GA' \Rightarrow \DE') \not\models \neg A$.

\item Suppose $\square^{c}_{\alpha} A \in \GA'$ and $\alpha \neq 0$.

Suppose $R_C((\GA' \Rightarrow \DE'), (\GA'' \Rightarrow \DE'')) = \beta$, and $\alpha \leqq \beta$.
If the reason for $\beta$ is (I), from $(\alpha, c) \preceq (\beta, c)$ and Step \ref{boxL} or \ref{boxLsym}, $A \in \GA''$.
If the reason for $\beta$ is (II), suppose $\beta = Suc(\beta')$.
Then, $(\alpha, c) \preceq (\beta', o)$ is established for the following reason.
If $(\beta', o) \prec (\alpha, c)$, then $Suc(\beta') < \alpha$ because $\alpha, \beta' \in (\GA_C \Rightarrow \DE_C, \mathcal{T}_C)_J$ and from the definitions of $\prec$ and $\mathbb{U}_C$, $\beta' < Suc(\beta') < \alpha$. In this case, $\beta < \alpha$, which is contrary to the assumption.
Therefore, from Step \ref{boxL} or \ref{boxLsym}, $A \in \GA''$.
If the reason for $\beta$ is (III), $\beta = 1$.
From Step \ref{boxLself}, $A \in \GA''$. 
From the inductive hypothesis, $(\GA'' \Rightarrow \DE'') \models A$ holds in all cases.
Therefore, $(\GA' \Rightarrow \DE') \models \square^{c}_{\alpha} A$.

\item Suppose $\square^{c}_{0} A \in \GA'$.

From Step \ref{boxC0}, $A \in \GA''$ for all $(\GA'' \Rightarrow \DE'') \in S_C$.
From the inductive hypothesis, $(\GA'' \Rightarrow \DE') \models A$ for all $(\GA'' \Rightarrow \DE'') \in S_C$. 
Therefore, $(\GA' \Rightarrow \DE') \models \square^{c}_{0} A$.

\item Suppose $\square^{o}_{\alpha} A \in \GA'$.

$\square^{o}_{1} A$ is always true because
there is no relation greater than $1$.

Suppose $\alpha \neq 1$, $R_C((\GA' \Rightarrow \DE'), (\GA'' \Rightarrow \DE'')) = \beta$, $\alpha < \beta$.
If the reason for $\beta$ is (I),
from $(\alpha, o) \preceq (\beta, c)$ and Step \ref{boxL} or \ref{boxLsym}, $A \in \GA''$.
If the reason for $\beta$ is (II), suppose $\beta = Suc(\beta')$.
From $\alpha, \beta' \in (\GA_C \Rightarrow \DE_C, \mathcal{T}_C)_J$, $\alpha < Suc(\beta')$, and the definitions of $\mathbb{U}_C$, $\alpha \leqq \beta'$.
From $(\alpha, o) \preceq (\beta', o)$ and Step \ref{boxL} or \ref{boxLsym}, $A \in \GA''$.
If the reason for $\beta$ is (III), $\beta = 1$.
From Step \ref{boxLself}, $A \in \GA''$. 

From the inductive hypothesis, $(\GA'' \Rightarrow \DE'') \models A$ holds in all cases.
Therefore, $(\GA' \Rightarrow \DE') \models \square^{o}_{\alpha} A$.


\item Suppose $\square^{c}_{\alpha} A \in \DE'$.

If $\alpha = 1$, from Step \ref{boxRself}, $A \in \DE'$.
From the inductive hypothesis, $(\GA' \Rightarrow \DE') \not\models A$.
If $\alpha \neq 1$, from Step \ref{boxR} and the definition of $R_C$, there exists $(\GA'' \Rightarrow \DE'') \in S_C$ such that $R_C((\GA' \Rightarrow \DE'), (\GA'' \Rightarrow \DE'')) = \alpha$ and $A \in \DE''$.
From the inductive hypothesis, $(\GA'' \Rightarrow \DE'') \not\models A$.
Therefore, $(\GA' \Rightarrow \DE') \not\models \square^{c}_{\alpha} A$.

\item Suppose $\square^{o}_{\alpha} A \in \DE'$.
$\alpha \neq 1$ because of the axiom, (wL), and (wR).
From Step \ref{boxR} and the definition of $R_C$, there exists $(\GA'' \Rightarrow \DE'') \in S_C$ such that $R_C((\GA' \Rightarrow \DE'), (\GA'' \Rightarrow \DE'')) = Suc(\alpha)$ and $A \in \DE''$.
From the inductive hypothesis, $(\GA'' \Rightarrow \DE'') \not\models A$.
Therefore, $(\GA' \Rightarrow \DE') \not\models \square^{o}_{\alpha} A$.
\end{itemize}
\end{proof}

\begin{lemma}\label{falseC} $\GA \Rightarrow \DE, \mathcal{T}$ is false in  $(S_C,R_C, P_C, V_C)$ under $\mathcal{E}_C$.
\end{lemma}
\begin{proof}
The corollary of Lemma \ref{incS} and Lemma \ref{fuTSMB}.
\end{proof}

\begin{theorem}[Completeness theorem for {\bf NSMB}]\label{compTSMB} If $\GA \Rightarrow \DE, \mathcal{T}$ is valid, then $\GA \Rightarrow \DE, \mathcal{T}$ is provable in {\bf NSMB}.
\end{theorem}
\begin{proof}
From Lemma \ref{falseC}, if $\GA \Rightarrow \DE, \mathcal{T}$ is not provable in {\bf NSMB}, there exists an {\bf MB}-realization

\noi
$(S_C,R_C,P_C,V_C)$ and an embedding $\mathcal{E}_C$
such that $\GA \Rightarrow \DE, \mathcal{T}$ is false under $\mathcal{E}_C$.
\end{proof}

\begin{theorem}[Cut-elimination theorem for {\bf NSMB}]\label{cutTSMB} If $\GA \Rightarrow \DE, \mathcal{T}$ is provable in {\bf NSMB}, there exists a proof of $\GA \Rightarrow \DE, \mathcal{T}$ that does not include the rule (cut).
\end{theorem}
\begin{proof}
The completeness theorem is proved without the rule (cut).
Therefore, the provability of a nested-sequent in {\bf NSMB}
does not depend on whether {\bf NSMB} contains (cut). 
\end{proof}

The construction of a canonical model stops within a finite number of steps.
The discussion does not change in essence if $J$ (and $\mathbb{U}$) is replaced by a suitable total ordered finite set instead of a set of real numbers.
Therefore, comparing $(\alpha, d)$ and $(\beta, d')$ can also be completed in a finite number of steps.

\begin{theorem}[Finite model property for {\bf MB}]\label{finiteMB} If $A$ is not valid, there exists an {\bf MB}-realization $(S,R,P,V)$ such that $S$ is a finite set and $A$ is not valid in it.
\end{theorem}
\begin{proof}
If $A$ is not valid, the above method could construct a finite canonical model of nested-sequent $ \Rightarrow A$. 
\end{proof}

\begin{theorem}\label{dTSMB} 
The validity problem for {\bf MB} is decidable.
\end{theorem}
\begin{proof}
The corollary of Theorem \ref{finiteMB}.
\end{proof}

\section{Nested-sequent calculus NSMB+}\label{MLMB+}

From a multi-relational frame point of view, $R$ in a {\bf MB}-realization is regarded as a set of binary relations with the {\it conditions} such as ``If there is a relation $\alpha$ from $s$ to $t$, then there is no relation $\beta$ ($\beta \neq \alpha$) from $s$ to $t$.''
In general, those binary relations are defined independently.
Some ingenuity is required to handle these conditions using formulas.
For example, the condition ``If there is a relation $R'$ from $s$ to $t$, then there is no other relation $R''$ from $s$ to $t$'' cannot be defined as a formula in standard modal logic. (Here, ``define'' has the same meaning as, for example, 
$\square p \to \square \square p$ defines the transitivity of a binary relation in a frame of modal logic.)
If the conditions of a frame cannot be defined as a formula, some problems may occur when proving the completeness theorem in a Hilbert-style system or a standard sequent system (see \cite{BlackBML} for these problems).
This problem does not occur in {\bf MB} because it only handles relational operators $\square^{c}_{\alpha}$ and $\square^{o}_{\alpha}$.
That is, the following ``normal'' modal symbols that correspond to only one modality are not included in {\bf MB} (other than $\square^{c}_{1} $).

\begin{description}
\item $V( \square^{=}_{\alpha} A)   = \{ s \in S  |$ for all $t \in S$, if $\alpha = R (s, t)$, then $t \in V(A) $ $\}$.
\end{description}

Relational operators make it simple to construct the canonical model. 
By employing only the maximum value among the numbers that satisfy a specific condition as a binary relation, we can have only one binary relation between any two possible worlds in the canonical model. (See \cite{extendQL} for concrete definitions. As mentioned briefly in the introduction, the completeness theorem of the Hilbert style system in \cite{extendQL} can be proved with this method if the infinite canonical model is acceptable.)
However, the above issue arises in a Hilbert-style system or a standard sequent system if $\square^{=}_{\alpha}$ is added to the language. Therefore, developed sequent becomes intrinsically important to adding $\square^{=}_{\alpha}$.

Adding $\square^{=}_{\alpha} A$ to the language of {\bf MB} and constructing a new logic is essential from both a physics and mathematical logic point of view since it broadens the range of expression.
Because $V(\square^{c}_{\alpha}A) = V(\square^{o}_{\alpha} A \wedge \square^{=}_{\alpha} A)$ holds, $\square^{c}_{\alpha}$ can be represented by $\square^{o}_{\alpha}$ and $\square^{=}_{\alpha}$, but $\square^{c}_{\alpha}$ and $\square^{o}_{\alpha}$ cannot represent $\square^{=}_{\alpha}$.
Therefore, it is desirable to define $\square^{c}_{\alpha}A$ as an abbreviation of $\square^{o}_{\alpha} A \wedge \square^{=}_{\alpha} A$ rather than a primitive formula.

Because $\square^{c}_{0}$ is a universal modality,
it is not directly related to $0$-relation, but $0$-relation is relevant to $\square^{=}_{0}$.
The definition (IV) of $R_C$ is inappropriate for $\square^{=}_{0}$ because (IV) is defined independently of occurrence of $\square^{=}_{0}A$ in the nested-sequent. Therefore, the truth of $\square^{=}_{0}A$ in the canonical model changes from intention, and the
proof of the completeness theorem fails. (Even if we add the concept of $0$-relation to embedding, the soundness of ($\square$R) will not be satisfied this time. It is currently unclear how this problem can be resolved if $\square^{=}_{0}$ is added.)
Therefore, we define the formulas of new logic {\bf MB+}
by removing all $\square^{c}_{\alpha}A \  (\alpha \neq 0)$ from the formulas of {\bf MB} and adding all $\square^{=}_{\alpha}A \  (0 < \alpha \leqq 1)$.

Basic definitions for {\bf MB+} are constructed as follows (but we only briefly describe the differences from the {\bf MB} case).
The relational symbols $^{d}_{\alpha}$ used in the modal symbols and the brackets 
in nested-sequent are $^{=}_{\alpha} (0 < \alpha \leqq 1)$, $^{o}_{\alpha} (0 \leqq \alpha \leqq 1)$, and $^{c}_{0}$.
The definition of embedding is changed by 
adding the following condition:

\begin{description}

\item If $(\GA_1 \Rightarrow \DE_1, [\GA_2 \Rightarrow \DE_2, \mathcal{T'}]^{=}_{\alpha}) \triangleleft (\GA \Rightarrow \DE, \mathcal{T})$ and $R((\mathcal{E}(\GA_1 \Rightarrow \DE_1),(\mathcal{E}(\GA_2 \Rightarrow \DE_2)) = \beta$, then $ \alpha = \beta$.


\end{description}

{\bf NSMB+} is defined by changing {\bf NSMB} as follows:

\begin{description}

\item

1. ($\square$ R self) is removed, and the following rule is added.

$$
\infer[\mbox{(= R self) }]{\SEQ{\| \GA}{\DE, \square^=_1 A, \mathcal{T} \|}}{
 \SEQ{\| \GA}{\DE, A, \mathcal{T} \|}}
\hspace{5mm}
$$

\item
2. The conditions  (1) and (2) in the annotation of {\bf NSMB} are changed as follows:

\begin{description}
\item (1) $d$ and $d'$ are $=$, and $\alpha = \beta$, or

$d$ and $d'$ are $o$, and $\alpha < \beta$, or

$d$ is $o$, $d'$ is $=$, and $\alpha < \beta$.  
\item (2) $d'$ is $=$ and $\alpha =1$, or

$d$ is $o$ and $\alpha \neq 1$.

\end{description}

\end{description}

\begin{theorem}[Soundness theorem for {\bf NSMB+}]\label{soTSMB+} If $\GA \Rightarrow \DE, \mathcal{T}$ is provable in {\bf NSMB+}, then
$\GA \Rightarrow \DE, \mathcal{T}$ is valid. 
\end{theorem}
\begin{proof}
Almost the same as Theorem \ref{soTSMB}.
\end{proof}

Some procedure for the composition of $\GA_C \Rightarrow \DE_C, \mathcal{T}_C$ is modified as follows:

\begin{enumerate}
\setcounter{enumi}{4}
\item\label{=L}
If $\| \GA' \Rightarrow \DE', [\GA'' \Rightarrow \DE'', \mathcal{T}'']^{d'}_{\beta}, \mathcal{T}' \| = \GA_{i} \Rightarrow \DE_{i}, \mathcal{T}_{i}$, $\alpha, \beta, d$ and $d'$ satisfy condition (1) of {\bf NSMB+}, and $\square^{d}_{\alpha} A \in \GA'$, then we construct $\GA_{i+1} \Rightarrow \DE_{i+1}, \mathcal{T}_{i+1} = \| \GA' \Rightarrow \DE', [A, \GA'' \Rightarrow \DE'', \mathcal{T}'']^{d'}_{\beta}, \mathcal{T}' \|$.  
This new nested-sequent is also not provable because of the rule ($\square$L).

\item\label{=Lsym}
If $\| \GA' \Rightarrow \DE', [\GA'' \Rightarrow \DE'', \mathcal{T}'']^{d'}_{\beta}, \mathcal{T}' \| = \GA_{i} \Rightarrow \DE_{i}, \mathcal{T}_{i}$, $\alpha, \beta, d$ and $d'$ satisfy condition (1) of {\bf NSMB+}, and $\square^{d}_{\alpha} A \in \GA''$, then we construct $\GA_{i+1} \Rightarrow \DE_{i+1}, \mathcal{T}_{i+1} = \| A, \GA' \Rightarrow \DE', [\GA'' \Rightarrow \DE'', \mathcal{T}'']^{d'}_{\beta}, \mathcal{T}' \|$.  
This new nested-sequent is also not provable because of the rule ($\square$ L sym).

\item\label{=Lself}
If $\| \GA' \Rightarrow \DE', \mathcal{T}' \| = \GA_{i} \Rightarrow \DE_{i}, \mathcal{T}_{i}$, $\alpha$ and $d$ satisfy condition (2) of {\bf NSMB+}, and $\square^{d}_{\alpha} A \in \GA'$, then we construct $\GA_{i+1} \Rightarrow \DE_{i+1}, \mathcal{T}_{i+1} = \| A, \GA' \Rightarrow \DE', \mathcal{T}' \|$.  
This new nested-sequent is also not provable because of the rule ($\square$ L self).

\setcounter{enumi}{8}
\item\label{=Rself}
If $\| \GA' \Rightarrow \DE', \mathcal{T}' \| = \GA_{i} \Rightarrow \DE_{i}, \mathcal{T}_{i}$, and $\square^{=}_{1} A \in \DE'$, then we construct $\GA_{i+1} \Rightarrow \DE_{i+1}, \mathcal{T}_{i+1} = \| \GA' \Rightarrow \DE', A, \mathcal{T}' \|$.  
This new nested-sequent is also not provable because of the rule (=R self).


\end{enumerate}

For the definition of $R_C$ of the canonical model, the following (I)' is added.

(I)' If $\GA' \Rightarrow \DE', [\GA'' \Rightarrow \DE'', \mathcal{T}'']^{=}_{\beta}, \mathcal{T}' \triangleleft \GA_C \Rightarrow \DE_C, \mathcal{T}_C$, then $R_C((\GA' \Rightarrow \DE'), (\GA'' \Rightarrow \DE'')) \stackrel{\mathrm{def}}{=} \beta$.


\begin{theorem}[Completeness theorem for {\bf NSMB+}]\label{compTSMB+} If $\GA \Rightarrow \DE, \mathcal{T}$ is valid, then $\GA \Rightarrow \DE, \mathcal{T}$ is provable in {\bf NSMB+}.
\end{theorem}
\begin{proof}

We change some parts of the proof of Lemma \ref{fuTSMB} as follows:

\begin{itemize}

\item Suppose $\square^{=}_{\alpha} A \in \GA'$ and $\alpha \neq 0$.

Suppose $R_C((\GA' \Rightarrow \DE'), (\GA'' \Rightarrow \DE'')) = \beta$, and $\alpha = \beta$.
If the reason for $\beta$ is (I)', from Step \ref{=L} or \ref{=Lsym}, $A \in \GA''$.
From the nature of $\mathbb{U}$ and $\alpha = \beta$, there is no case where (II) is the reason for $\beta$.


\item Suppose $\square^{o}_{\alpha} A \in \GA'$.

Suppose $\alpha \neq 1$, $R_C((\GA' \Rightarrow \DE'), (\GA'' \Rightarrow \DE'')) = \beta$, $\alpha < \beta$.

If the reason for $\beta$ is (I)',
from Step \ref{=L} or \ref{=Lsym}, $A \in \GA''$.
If the reason for $\beta$ is (II), suppose $\beta = Suc(\beta')$.
From $\alpha, \beta' \in (\GA_C \Rightarrow \DE_C, \mathcal{T}_C)_J$, $\alpha < Suc(\beta')$, and the definitions of $\mathbb{U}_C$, $\alpha \leqq \beta'$.
From $\alpha < \beta'$ and Step \ref{=L} or \ref{=Lsym}, $A \in \GA''$.
If the reason for $\beta$ is (III), $\beta = 1$.
From Step \ref{boxLself}, $A \in \GA''$. 



\item Suppose $\square^{=}_{\alpha} A \in \DE'$.

If $\alpha = 1$, from Step \ref{boxRself}, $A \in \DE'$.
From the inductive hypothesis, $(\GA' \Rightarrow \DE') \not\models A$.
If $\alpha \neq 1$, from Step \ref{boxR} and the definition of $R_C$, there exists $(\GA'' \Rightarrow \DE'') \in S_C$ such that $R_C((\GA' \Rightarrow \DE'), (\GA'' \Rightarrow \DE'')) = \alpha$ and $A \in \DE''$.
\end{itemize}  \end{proof}

The following theorems can also be proved in the same way as the {\bf NSMB} case.

\begin{theorem}[Cut-elimination theorem for {\bf NSMB+}]\label{cutTSMB+} If $\GA \Rightarrow \DE, \mathcal{T}$ is provable in {\bf NSMB+}, there exists a proof of $\GA \Rightarrow \DE, \mathcal{T}$ that does note include the rule (cut).
\end{theorem}

\begin{theorem}[Finite model property for {\bf MB+}]\label{finiteMB+} If $A$ is not a valid formula of {\bf MB+}, there exists an {\bf MB}-realization $(S,R,P,V)$ such that $S$ is a finite set and $A$ is not valid in it.
\end{theorem}

\begin{theorem}\label{dTSMB+} 
The validity problem for {\bf MB+} is decidable.
\end{theorem}

The definition of interpretation $\tau$ is the same as for {\bf MB} (except that $d$ could be $=$). 

\begin{theorem}

$\GA \Rightarrow \DE, \mathcal{T}$ is valid iff $\tau(\GA \Rightarrow \DE, \mathcal{T})$ is valid.

\end{theorem}

\section*{Acknowledgements}

\noi
This work was supported by JSPS KAKENHI Grant Number 20K19740.

\nocite{*}
\bibliographystyle{eptcs}
\bibliography{generic}

\end{document}